\documentclass[12pt]{article}

\usepackage[utf8]{inputenc}

\usepackage{amsmath,amsthm,amssymb}
\usepackage{mathrsfs}
\usepackage{authblk}

\theoremstyle{plain}
\newtheorem{theorem}{Theorem}
\newtheorem{lemma}{Lemma}
\newtheorem{proposition}{Proposition}
\theoremstyle{definition}
\newtheorem{definition}{Definition}
\newtheorem{corollary}{Corollary}
\theoremstyle{remark}
\newtheorem{remark}{Remark}
\newtheorem*{remark*}{Remark}

\numberwithin{equation}{section}

\usepackage[margin=25mm]{geometry}

\usepackage{fancyhdr}
\begin{document}

\title{Extremal positive maps on $M_{3}(\mathbb{C})$ \\ and idempotent matrices}

\author[1]{Marek Miller\thanks{marek.miller@ift.uni.wroc.pl}}
\author[1]{Robert Olkiewicz\thanks{robert.olkiewicz@ift.uni.wroc.pl}}
\affil[1]{Institute of Theoretical Physics, Uniwersytet Wroc{\l}awski, Poland}

\date{\today}

\maketitle

\abstract{
A new method of analysing positive bistochastic maps on the algebra of complex
matrices $M_{3}$ has been proposed.
By identifying the set of such maps with a convex set of linear operators on
$\mathbb{R}^{8}$, one can employ techniques from the theory of
compact semigroups to obtain results concerning asymptotic properties
of positive maps.
It turns out that the idempotent elements play a crucial role
in classifying the convex set into subsets, in which representations of
extremal positive maps are to be found.
It has been show that all positive bistochastic maps,
extremal in the set of all positive maps of $M_{3}$,
that are not Jordan isomorphisms of $M_3$ are represented by matrices that fall into two possible
categories, determined by the simplest idempotent matrices:
one by the zero matrix, and the other by a one dimensional orthogonal projection.
Some norm conditions for matrices representing possible extremal maps
have been specified and examples of maps from both categories have been
brought up, based on the results published previously.

\vspace{0.1cm}

}

\section*{Introduction}
\label{sec:Introduction}

\paragraph{}
Positive maps of operator algebras constitute a rich area of research, directly
connected to the theory of quantum entanglement.
In 1990s,
A.\,Peres and P.\,M.\,R. {Horo-deckis} pointed out at the intrinsic relation
between separable states of composite quantum systems
and positive maps of algebras of observables
\cite{peres1996separability,horodecki1996separability}.
The well established criterion of separability, proposed in the mentioned papers,
reveals a one-to-one correspondence between
positive maps and entanglement witnesses
\cite{chruscinski2014entanglement}.
This \emph{Peres-Horodecki criterion},
originally proposed for maps on algebras $M_{n}$ of square complex matrices of size $n$,
which holds true even in the most general setting of
injective von Neumann algebras \cite{miller2014horodeckis},
is computationally feasible as long as the structure of general positive
maps on operator algebras representing a composite quantum system in
question is known.
To this day, the complete characterisation
of positive maps have been obtained only for the algebra $M_{2}$
and the maps between $M_{2}$ and $M_{3}$
\cite{stormer1963positive,woronowicz1976positive}.

To analyse the structure of positive maps in the simplest, still unresolved
case of maps on $M_{3}$,
we propose a continuation of the reasoning conducted in our previous paper
\cite{miller2015stable}.
We have established a connection between positive maps that preserve both
the trace of matrices and the identity matrix, the so-called bistochastic
maps, and their stable subspaces that have the structure of Jordan algebras.
Here, we go one step further and explore the relation between those stable
subspaces and the idempotent real matrices that represent the conditional
expectations projecting onto the spaces.
To this end, we employ mostly geometrical techniques that allowed us before to
establish the structure theorem for maps on the algebra $M_{2}$
\cite{miller2015topology},
as well as the methods from the theory of compact semigroups
\cite{schwarz1955hausdorff,chow1975compact}.

The main result of this paper makes possible to outline a program,
suggesting where the extremal positive maps on $M_{3}$ are to be found
with respect to their connection to associated idempotents
(see Theorem \ref{thm:LowerIndices} below).
We have shown that all positive bistochastic maps,
extremal in the set of all positive maps of $M_{3}$,
that are not
Jordan isomorphisms of $M_3$ are represented by matrices that fall into two possible
categories, determined by the simplest idempotent matrices:
one by the zero matrix, and the other by a one dimensional orthogonal projection.
As a corollary, we specify some norm conditions for matrices representing possible
extremal maps.
The structure of the paper concentrates on building a mathematical framework
necessary to prove the final result.
We start with a brief listing of necessary notation and definitions.


\section{Preliminaries}
\label{sec:Preliminaries}

\paragraph{}
Let $n, k \in \mathbb{N}$ be grater than 0.
Let $M_{n} = M_{n}(\mathbb{C})$ denote the algebra of complex square
matrices of size $n$.
The algebra of real matrices will always be denoted explicitly by
$M_{n}(\mathbb{R})$.
For $A \in M_{n}$, the norm $||A||$ is understood to be the standard operator
norm, i.e. the maximal singular value of the matrix $A$.
For the Hilbert-Schmidt norm (HS-norm) of $A$,
we reserve the symbol
$||A||_{HS} = \sqrt{\text{tr} A^{*} A}$,
where $\text{tr}$ denotes the trace operation,
and $A^{*}$ is the conjugate-transpose of $A$.
The HS-norm of $A$ can be computed as the sum of squares of
singular values of $A$.
The Hilbert-Schmidt inner product, induced by the HS-norm,
is defined as
$\langle A, B \rangle_{HS} = \text{tr} A^{*} B$,
for $A, B \in M_{n}$.
The identity matrix of $M_{n}$, we denote by $\mathbf{1}_{n}$,
or simply $\mathbf{1}$, and the null matrix by $\hat{0}$.

We say that a matrix $A \in M_{n}$ is positive-semidefinite,
or simply \emph{positive},
if the inner product $\langle \eta, A \eta \rangle \geq 0$,
for any vector $\eta \in \mathbb{C}^{n}$
(i.e. $A = A^{*}$ and $A$ has a non-negative spectrum).
A linear map $S\!: M_{n} \rightarrow M_{n}$ is said to be positive,
indicated: $S \geq 0$,
if for any $A \in M_{n}$ such that $A \geq 0$,
we have $S(A) \geq 0$.
The operator norm of the linear map $S$ is given by
\begin{equation}
    \label{eq:SupremumNormOfS}
    ||S|| = \sup \left\{ \, ||S(A)||, \,\, A \in M_{n}, \, ||A|| = 1 \right\}.
\end{equation}
It is true that if $S$ is positive, then
$||S|| = ||S(\mathbf{1})||$.
Any positive map is Hermitian,
i.e.  $S(A^{*}) = S(A)^{*}$, for all $A \in M_{n}$.
The identity map of $M_{n}$ is labelled $I_{n}$,
or simply $I$.
The convex cone of all positive maps of $M_{n}$ is denoted
by $\mathcal{P}(M_{n})$.
A positive map $S$ is extremal,
if for any positive map $T: M_{n} \rightarrow M_{n}$ such that
$S - T \in \mathcal{P}(M_{n})$,
i.e. $0 \leq T \leq S$,
we have $T = \alpha S$ for some number $0 \leq \alpha \leq 1$.
It is true that every positive map
can be written as a convex combination of extremal ones.
If  $S\!: M_{n} \rightarrow M_{n}$ is a positive map such that
$S(\mathbf{1}) = \mathbf{1}$ and
$\text{tr} S(A) = \text{tr} A$,
for any $A \in M_{n}$,
then we call it \emph{bistochastic},
or \emph{doubly stochastic}.

From now on,
let us fix $n = 3$.
We choose the set of normalised Gell-Mann matrices,
$\left \{ \lambda_{\mu} \right \}_{\mu=0}^{8}$, by taking:
\begin{equation}
\nonumber
    \lambda_{1} = \frac{1}{\sqrt{2}} \begin{pmatrix}
            0 & 1 & 0 \\
            1 & 0 & 0 \\
            0 & 0 & 0
        \end{pmatrix}, \quad
    \lambda_{2} = \frac{1}{\sqrt{2}} \begin{pmatrix}
            0 & -i & 0 \\
            i & 0 & 0 \\
            0 & 0 & 0
        \end{pmatrix}, \quad
    \lambda_{3} = \frac{1}{\sqrt{2}} \begin{pmatrix}
            1 & 0 & 0 \\
            0 & -1 & 0 \\
            0 & 0 & 0
        \end{pmatrix},
\end{equation}
\begin{equation}
\nonumber
    \lambda_{4} = \frac{1}{\sqrt{2}} \begin{pmatrix}
            0 & 0 & 1 \\
            0 & 0 & 0 \\
            1 & 0 & 0
        \end{pmatrix}, \quad
    \lambda_{5} = \frac{1}{\sqrt{2}} \begin{pmatrix}
            0 & 0 & -i \\
            0 & 0 & 0 \\
            i & 0 & 0
        \end{pmatrix}, \quad
    \lambda_{6} = \frac{1}{\sqrt{2}} \begin{pmatrix}
            0 & 0 & 0 \\
            0 & 0 & 1 \\
            0 & 1 & 0
        \end{pmatrix},
\end{equation}
\begin{equation}
\nonumber
    \lambda_{7} = \frac{1}{\sqrt{2}} \begin{pmatrix}
            0 & 0 & 0 \\
            0 & 0 & -i \\
            0 & i & 0
        \end{pmatrix}, \quad
    \lambda_{8} = \frac{1}{\sqrt{6}} \begin{pmatrix}
            1 & 0 & 0 \\
            0 & 1 & 0 \\
            0 & 0 & -2
        \end{pmatrix}
\end{equation}
and $\lambda_{0} = \frac{1}{\sqrt{3}} \mathbf{1}_{3}$;
which is an orthonormal basis for $M_{3}$ with respect to the HS inner product.
Any self-adjoint matrix $A = A^{*} \in M_{3}$ can be written as
$A = \lambda(a)$, where
$\lambda(a) = \sum_{\mu=0}^{8} a_{\mu} \lambda_{\mu}$,
and
$a = (a_{0}, a_{1}, \ldots, a_{8}) \in \mathbb{R}^{9}$.
Sometimes, we will use a simplified notation:
$a = (a_{0}, \vec{a})$,
where $\vec{a} = (a_{1},\ldots,a_{8}) \in \mathbb{R}^{8}$,
and $\lambda(a) = \lambda(a_{0}, \vec{a}) = a_{0} \lambda_{0} + \vec{a} \cdot \vec{\lambda}$.
For $a = (a_{0}, \vec{a})$ and $b = (b_{0}, \vec{b})$,
we have
$\langle \lambda(a), \lambda(b) \rangle_{HS} = a_{0} b_{0} + \vec{a} \cdot \vec{b}$.
For a bistochastic map $S$,
let us define a matrix $x \in M_{8}(\mathbb{R})$ by
$x_{ij} = \text{tr} \lambda_{i} S(\lambda_{j})$,
$i,j = 1,2,\ldots,8$.
Then it is easy to see that $S$ acts on a self-adjoint matrix by
\begin{equation}
\label{eq:isomoprhism}
S(\lambda(a_{0}, \vec{a})) = \lambda(a_{0}, x \vec{a}).
\end{equation}
For a matrix $x \in M_{8}(\mathbb{R})$,
we denote a linear map on $M_{3}$ that preserves both the identity and trace,
defined by the relation \eqref{eq:isomoprhism},
by $S_{x}$.
Let $\Lambda \subset M_{8}(\mathbb{R})$ be a set of those real matrices $x$,
for which $S_{x}$ is a bistochastic map:
$\Lambda = \left \{ x \in M_{8}(\mathbb{R}) \: |  \: S_{x} \geq 0 \right \}$.
Therefore, there is a one-to-one correspondence between $\Lambda$
and the set of bistochastic maps on $M_{3}$.
The mapping $S_{x} \mapsto x$ is in fact a semigroup isomorphism,
for which $S_{x}^{*} = S_{x^{t}}$,
where generally $S^{*}$ denotes the HS-adjoint map of $S$,
given by:
$\text{tr} \, S^{*}(A) B = \text{tr} A S(B)$,
for all $A, B \in M_{3}$.
The structure of the set $\Lambda$ is fairly complicated.
First, it is a closed convex set, i.e. for any $x,y \in \Lambda$,
$\lambda x + (1-\lambda)y \in \Lambda$,
for $0 \leq \lambda \leq 1$.
Moreover, is is a compact, convex and real topological semigroup
(see e.g. \cite{schwarz1955hausdorff,chow1975compact} for the definitions)
with involution being the matrix transposition: $A \mapsto A^{t}$.
Obviously,
$S_{\mathbf{1}_{8}} = I$.
The structure of the set analogous to $\Lambda$,
but representing maps on the algebra $M_{2}$,
has been studied using geometrical methods in \cite{miller2015topology}.
Because there is no such geometrical identification for $n = 3$,
this time we exploit the semigroup aspect of the set $\Lambda$.

\begin{proposition}
\label{prop:LambdaAndBalls}
    Let $\overline{K}_{r}(\hat{0})$ denote a closed
    ball in $M_{8}(\mathbb{R})$ with respect to the operator norm,
    centred around $\hat{0}$, with radius $r > 0$.
    Then
    \begin{equation}
        \overline{K}_{\frac{1}{2}}(\hat{0}) \subset \Lambda \subset \overline{K}_{1}(\hat{0}).
    \end{equation}
\end{proposition}
\begin{proof}
    Suppose that $x \in M_{8}(\mathbb{R})$ and $||x|| \leq \frac{1}{2}$.
    Let $\vec{m}, \vec{n} \in \mathbb{R}^{8}$ be such that
    $P_{\vec{m}} = \lambda(\frac{1}{\sqrt{3}}, \vec{n})$ and
    $P_{\vec{n}} = \lambda(\frac{1}{\sqrt{3}}, \vec{m})$
    are orthogonal projections in $M_{3}$.
    It is easy to check that in that case
    $|| \vec{m} ||_{2} = || \vec{n} ||_{2} = \sqrt{\frac{2}{3}}$,
    where $|| \cdot||_{2}$ denotes
    the standard Euclidean norm.
    Because $||x|| \leq \frac{1}{2}$, we have that
    \begin{equation}
        \text{tr} \, P_{\vec{m}} S_{x}(P_{\vec{n}}) =
        \frac{1}{3} + \langle \vec{m}, x \vec{n} \rangle \geq
        \frac{1}{3} - |\langle \vec{m}, x \vec{n} \rangle| \geq 0.
    \end{equation}
    Hence $S_{x}$ is positive, which means that $x \in \Lambda$.

    Let now $x \in \Lambda$.
    Then for any $A = A^{*} \in M_{3}$, $S_{x}$ fulfils the
    \emph{Kadison-Schwarz inequality}
    \cite{choi1980some}:
    \begin{equation}
     \label{eq:KSineq}
        S(A)^{2} \leq S(A^{2}).
    \end{equation}
    Let $\vec{n} \in \mathbb{R}^{8}$; since $\lambda(\frac{1}{\sqrt{3}}, \vec{n})$
    is self-adjoint, by \eqref{eq:KSineq} and by the fact that
    $S_{x}$ preserves trace,
    we have
    \begin{equation}
        \frac{1}{3} + || x \vec{n} ||_{2}^{2} = \text{tr}\,
        \left ( S_{x}(\lambda(\frac{1}{\sqrt{3}}, \vec{n}) \right)^{2} \leq
        \text{tr}\, S \left( \lambda(\frac{1}{\sqrt{3}}, \vec{n})^{2} \right) =
        \text{tr}\, \lambda(\frac{1}{\sqrt{3}}, \vec{n})^{2} =
        \frac{1}{3} + ||\vec{n}||_{2}^{2}.
    \end{equation}
Hence, $||x\vec{n}||_{2} \leq ||\vec{n}||_{2}$, for any $\vec{n} \in \mathbb{R}^{8}$,
which means that $x \in \overline{K}_{1}(\hat{0})$.
\end{proof}

Let $P_{8} \in M_{8}(\mathbb{R})$ denote a diagonal matrix
$P_{8} = \text{diag}(0,0,0,0,0,0,0,1)$;
and so on for other sets of indices:
$P_{38} = \text{diag}(0,0,1,0,0,0,0,1)$, etc.
It is easy to check that all matrices:
$P_{8}, P_{38}, P_{138}, P_{1238}, P_{13468}$ are idempotent elements
belonging to $\Lambda$.
In addition, there are no idempotent elements in $\Lambda$ of rank 6 or 7,
and the only idempotent of rank 8 is the identity $\mathbf{1}_{8}$.
That fact will become clear in the light of the proof of
Theorem \ref{thm:Idempotents} below (see Remark \ref{rem:noRank6or7}).
Let also $G_{3} = \text{Ad} \, \text{SU}(3) \subset \Lambda$ denotes
the group of those matrices $g \in \Lambda$ such that
$S_{g}$ is an automorphism: $S_{g}(A) = U A U^{*}$,
for a unitary matrix $U \in \text{SU}(3)$ and for any $A \in M_{3}$.
It is evident that $G_{3} \subset \text{SO}(8)$,
the special group of orthogonal matrices.
The set of those $x \in \Lambda$ such that $S_{x}$ is a bistochastic map,
extremal in the set
of all positive maps on $M_{3}$ (not necessarily bistochastic),
is denoted by
$\text{Ext}_{0}(\Lambda)$,
whereas the set of extremal points of the convex set $\Lambda$ is labelled as
$\text{Ext}(\Lambda)$.
The matrix $x \in \text{Ext}(\Lambda)$,
if and only if for all
$y_{1}, y_{2} \in \Lambda$ and $0 \leq \lambda \leq 1$,
if $x = \lambda y_{1} + (1-\lambda) y_{2}$,
then $y_{1} = y_{2} = x$.
It is true that $\text{Ext}_{0}(\Lambda) \subset \text{Ext}(\Lambda)$.
The next fact follows from Proposition \ref{prop:LambdaAndBalls}.

\begin{proposition}
    \label{prop:oneHalfofOrthogonal}
    Let $x \in \mathrm{Ext}(\Lambda)$.
    Suppose that $||x|| = \frac{1}{2}$;
    then $ x = \frac{1}{2} R$,  for $R \in \mathrm{O}(8)$.
\end{proposition}
\begin{proof}
    Let $x \in \mathrm{Ext}(\Lambda)$ and $||x|| = \frac{1}{2}$.
    Let $x = R |x|$ be the polar decomposition of $x$,
    $R \in \text{O}(8)$, $|x| = \sqrt{x^{t} x}$.
    Suppose that $|x| \neq \frac{1}{2} \mathbf{1}_{8}$.
    Let $y_{1} = \frac{1}{2} \mathbf{1}_{8}$ and
    $y_{2} = 2 |x| - \frac{1}{2} \mathbf{1}_{8}$.
    Then $- \frac{1}{2} \mathbf{1}_{8} \leq y_{2} \leq \frac{1}{2} \mathbf{1}_{8}$,
    so $|| y_{2} || \leq \frac{1}{2}$
    and by Proposition \ref{prop:LambdaAndBalls},
    we have that both $R y_{1}, R y_{2} \in \Lambda$.
    Then $x = \frac{1}{2} R y_{1} + \frac{1}{2} R y_{2}$,
    and thus it cannot be extremal, a contradiction.
    Hence, $x = \frac{1}{2} R$.
\end{proof}

\begin{remark}
    By Proposition \ref{prop:LambdaAndBalls},
    we see that $x \in \text{Ext}(\Lambda)$ implies that
    $||x|| \geq \frac{1}{2}$.
    Therefore, $x = \frac{1}{2} R$, $R \in \mathrm{O}(8)$,
    are the only possible elements in $\text{Ext}(\Lambda)$,
    and so in $\text{Ext}_{0}(\Lambda)$,
    with the norm $||x|| = \frac{1}{2}$.
\end{remark}


\section{Idempotent and extremal elements of $\Lambda$}
\label{sec:main}

Let $x \in \Lambda$ and let
$\langle x \rangle \subset \Lambda$ be the semigroup generated by $x$:
$\langle x \rangle = \left \{ x^{k} \: | \: k \in \mathbb{N}, k \geq 1 \right \}$.
By $\overline{\langle x \rangle}$,
we denote the closure of $\langle x \rangle$ in $M_{8}(\mathbb{R})$.
The proof of the following proposition is presented in
    \cite[Lemma 3]{schwarz1955hausdorff}.

\begin{proposition}
    \label{prop:UniqeClusterPoint}
    Since $\Lambda$ is closed, $\overline{\langle x \rangle} \subset \Lambda$.
    The set $\overline{\langle x \rangle}$ contains a unique idempotent,
        denoted by $e_{x}$.
\end{proposition}

\begin{definition}
    For the set $\Lambda$, we define the following subsets:
    \begin{enumerate}
       \item the set of idempotents of $\Lambda$:
            $\mathcal{E}(\Lambda) = \left \{ e \in \Lambda \: |  \: e^{2} = e \right \}$;
       \item the set of nilpotent elements:
            $\mathcal{N}(\Lambda) = \left \{ x \in \Lambda \: |  \: \lim \limits_{n \rightarrow \infty} x^{n} = \hat{0} \right \}$;
       \item the group of invertible elements:
            $\mathcal{G}(\Lambda) = \left \{ x \in \Lambda \: |  \: \exists \, y \in \Lambda, xy = yx = \mathbf{1}_{8} \right \}$.
    \end{enumerate}
   For an idempotent element $e \in \mathcal{E}(\Lambda)$,
   we define the following subsets of $\Lambda$:
   \begin{enumerate}
        \item let $H(e)$ be the maximal subgroup of $\Lambda$
            containing $e \in \mathcal{E}(\Lambda)$;
        \item $Q(e) = \left \{ x \in \Lambda \: | \: e_{x} = e \right \}$,
        where $e_{x}$ is the unique idempotent element of $\overline{\langle x \rangle}$.
   \end{enumerate}
\end{definition}

\begin{remark}
\label{rem:GOfLambda}
It is a known fact from the semigroup theory that for each idempotent
element of a semigroup, there is exactly one maximal subgroup containing it.
In particular, in our case of the semigroup $\Lambda$,
$H(\hat{0}) = \left \{ \hat{0} \right \}$,
$H(\mathbf{1}_{8}) = \mathcal{G}(\Lambda)$,
obviously,
and moreover $H(P_{8}) = \left \{ P_{8} \right \}$.
Indeed,
if the matrix elements of $x$ are $(x_{ij})_{i,j=1}^{8}$,
then from the fact that $x \in H(P_{8})$, follows
$x = x P_{8} = P_{8} x$
($x$ belongs to the maximal group containing $P_{8}$,
and the idempotent $P_{8}$ is the identity for that group).
Hence, $x = P_{8} \, x P_{8} = x_{88} P_{8}$.
Because there is a sequence of natural numbers $n_{k}$, $k=1,2,3,\ldots$,
such that $x^{n_{k}} \overset{k}{\rightarrow} P_{8}$, so
$|x_{88}| = 1$. But $- P_{8} \notin \Lambda$, and thus $x = P_{8}$.
\end{remark}

The following proposition, proven in general terms in
\cite[Theorem 8]{schwarz1955hausdorff},
will be repeatedly used in the subsequent reasoning.
\begin{proposition}
    \label{prop:SchwarzSemigroups}
    Let $e \in  \mathcal{E}(\Lambda)$.
    Then $H(e) = e \, Q(e) = Q(e) \, e$.
\end{proposition}
\begin{proof}
    Observe that the semigroup $\Lambda$ fulfils the assumptions
    of \cite[Theorem 8]{schwarz1955hausdorff}.
\end{proof}

\begin{remark}
    \label{rem:Qof1}
It is true that $Q(\hat{0}) = \mathcal{N}(\Lambda)$.
Indeed, let $x \in Q(\hat{0})$,
then by definition, there is a sequence of natural numbers $n_{k}$ such that
$x^{n_{k}} \overset{k}{\rightarrow} \hat{0}$.
For some fixed $k$, for any $l \geq n_{k}$,
since $||x|| \leq 1 $
(see Proposition \ref{prop:LambdaAndBalls}),
one has
$||x^{l}|| \leq ||x^{l - n_{k}}|| \, ||x^{n_{k}}|| \leq ||x^{n_{k}}||$.
Hence $x^{l} \overset{l}{\rightarrow} \hat{0}$
and $Q(\hat{0}) \subset \mathcal{N}(\Lambda)$.
The reverse inclusion is obvious.
Moreover, since we have that $H(\mathbf{1}_{8}) = \mathcal{G}(\Lambda)$,
    by Proposition \ref{prop:SchwarzSemigroups},
    $H(\mathbf{1}_{8}) = Q(\mathbf{1}_{8}) \cdot \mathbf{1}_{8}$,
    and hence $Q(\mathbf{1}_{8}) = \mathcal{G}(\Lambda)$.
\end{remark}

Our task is to analyse the family of sets
$\left \{ Q(e), e \in \mathcal{E}(\Lambda) \right \}$,
and in particular,
to show, where among this family the elements of
$\text{Ext}_{0}(\Lambda)$ are to be found.
To this end, we present the following series of results.

\begin{proposition}
    For $e_{1}, e_{2} \in \mathcal{E}(\Lambda)$,
    $e_{1} \neq e_{2}$,
    the sets $Q(e_{1})$ and $Q(e_{2})$ are disjoint.
    Moreover,
    $\Lambda = \! \bigcup \limits_{e \in \mathcal{E}(\Lambda)} \! Q(e)$.
\end{proposition}
\begin{proof}
    If $e_{1} \neq e_{2}$, by Proposition \ref{prop:UniqeClusterPoint},
    the sets $Q(e_{1})$ and $Q(e_{2})$ must be disjoint.
    For $x \in \Lambda$, we have that $x \in Q(e_{x}) \subset \Lambda$,
    and the assertion follows.
\end{proof}

It should be noted that for $g \in \mathcal{G}(\Lambda)$,
because $||g|| \leq 1$ (see Proposition \ref{prop:LambdaAndBalls}),
and $g^{-1}  \in \Lambda$, then $|| g || = 1$.
Thus, $g$ is an orthogonal matrix.
In particular,
for any $e \in \mathcal{E}(\Lambda)$,
$g e g^{t} \in \mathcal{E}(\Lambda)$.

\begin{lemma}
    \label{lem:equivClassesOfQe}
    Let $e \in \mathcal{E}(\Lambda)$ and $g \in \mathcal{G}(\Lambda)$.
    Then $Q(g e g^{t}) = g Q(e) g^{t} =
    \left \{ g x g^{t} \: | \: x \in Q(e) \right \}$.
\end{lemma}
\begin{proof}
    It is clear that $g \langle x \rangle g^{t} = \langle g x g^{t} \rangle$.
    Because the mapping $x \mapsto g x g^{t}$ is a homeomorphism,
    so $g \overline{\langle x \rangle} g^{t} = \overline{\langle g x g^{t} \rangle}$.
    Let $x \in Q(g e g^{t})$.
    Then $g e g^{t} \in \overline{\langle x \rangle}$.
    Hence $e \in g^{t} \overline{\langle x \rangle} g = \overline{\langle g^{t} x g \rangle}$.
    This in turn means that $g^{t} x g \in Q(e)$, i.e.
    $x \in g Q(e) g^{t}$, which establishes that
    $Q(g e g^{t}) \subset g Q(e) g^{t}$.
    The reverse inclusion follows from an analogous reasoning.
\end{proof}

Next, we describe the structure of the set of idempotents
in $\Lambda$.
For $e \in \mathcal{E}(\Lambda)$,
let $\mathcal{E}_{G_{3}}(e) = \left \{ g e g^{t} \: | \: g \in G_{3} \right \}$.
At first, in the following lemma, we recall a known fact
that an idempotent contractive operator on a Hilbert space
is an orthogonal projection
(see e.g. \cite[Problem 5.3.14]{abramovich2002problems}).
\begin{lemma}
    \label{lem:eIsProj}
    Let $e \in \mathcal{E}(\Lambda)$, then $e^{t} = e$, and hence,
    $e$ is an orthogonal projection in $M_{8}(\mathbb{R})$.
\end{lemma}
\begin{proof}
    Because $e^{2} = e$, so $||e|| = ||e^{2}|| \leq ||e||^{2}$,
    hence $||e|| \geq 1$.
    Since $e \in \Lambda$, by Proposition \ref{prop:LambdaAndBalls}, $||e|| =1$.
    Suppose that $\vec{n} \in \ker e$ and $\vec{m} \in \mathcal{R}(e)$,
    the range of the operator $e$.
    Let $\alpha \in \mathbb{R}$.
    Then
\begin{equation}
    || \alpha \vec{m} ||^{2} = || P (\vec{n} + \alpha \vec{m}) ||^{2} \leq
    || \vec{n} + \alpha \vec{m}||^{2} \leq ||\vec{n}||^{2} +
        2 \alpha \langle \vec{n}, \vec{m} \rangle + || \alpha \vec{m}||^{2},
\end{equation}
i.e. $||\vec{n}||^{2} + 2 \alpha  \langle \vec{n}, \vec{m} \rangle \geq 0$,
for any $\alpha \in \mathbb{R}$.
It means that
$\langle \vec{n}, \vec{m} \rangle = 0$,
for every  $\vec{n} \in \ker e$ and $\vec{m} \in \mathcal{R}(e)$,
i.e.
$\ker e \perp \mathcal{R}(e)$, which proves that $e = e^{t}$.
\end{proof}

\begin{proposition}
\label{prop:htranspose}
Let $e \in \mathcal{E}(\Lambda)$ and $h \in H(e)$.
Then $h^{t} \in H(e)$, and $h^{t} h = h h^{t} = e$.
\end{proposition}
\begin{proof}
    By Proposition \ref{prop:SchwarzSemigroups},
    $H(e) = Q(e) e = e \, Q(e)$.
    If $x \in Q(e)$, then $x^{t} \in Q(e)$,
    by Lemma \ref{lem:eIsProj}.
    Hence, $h^{t} \in H(e)$.
    Since $e$ is a projection, $e \leq \mathbf{1}_{8}$;
    moreover $h^{t} h \leq \mathbf{1}_{8}$, because $h \in \Lambda$.
    Thus, $\hat{0} \leq (h^{t} h)^{k} \leq h^{t} h = e h^{t} h e \leq e$,
    for any $k \in \mathbb{N}$.
    If $n_{k}$ is a sequence of natural numbers such that
    $(h^{t} h)^{n_{k}} \overset{k}{\rightarrow} e$,
    we obtain that $h^{t} h = e$.
    By the same reasoning, we have also: $h h^{t} = e$.
\end{proof}

Let us recall that for a bistochastic map $S$, by
$K_{S}$ we denote the stable subspace of $S$ defined by
(see \cite[Eq. (3.3)]{miller2015stable}):
\begin{equation}
    K_{S} = \left \{ x \in M_{3} \:|\:
            \forall k \in \mathbb{N} \,\,
            || S^{k} x ||_{HS} = || S^{*k} x ||_{HS} =  ||x||_{HS}
    \right \},
\end{equation}
where $S^{*}$ is the adjoint map of $S$.
The fact that $H(\mathbf{1}_{8}) = \mathcal{G}(\Lambda)$
can be generalised to the following result.
\begin{theorem}
    Let $e \in \mathcal{E}(\Lambda)$ and
    $K_{S_{e}}$ be the stable subspace of the map $S_{e}$.
    Then $K_{S_{e}}$ is a Jordan algebra and
    $H(e) \cong \text{Aut}_{J} \, K_{S_{e}}$,
    the group of Jordan automorphisms of $K_{S_{e}}$.
\end{theorem}
\begin{proof}
    Because $e$ is idempotent, by
    \cite[Corollary 3]{miller2015stable},
    the space $K_{S_{e}} = S_{e}(M_{3})$ is a Jordan algebra.
    The map $S_{e}$ is in fact the conditional expectation onto $K_{S_{e}}$.
    Let $h \in H(e)$.
    By Proposition \ref{prop:htranspose}, for any $k\in \mathbb{N}$,
    $S_{h}^{* k} S_{h}^{k} = S_{h^{t}}^{k} S_{h}^{k} =  S_{(h^{t})^{k} h^{k}} = S_{e}$,
    and the same for $S_{h}^{k} S_{h}^{* k}= S_{e}$.
    Thus, the stable algebra $K_{S_{h}} = K_{S_{e}}$
    (compare \cite{miller2015stable}, eq. (3.3) and below).
    Again, by \cite[Corollary 3]{miller2015stable},
    $\varphi_{h} = S_{h} \big |_{K_{S_{e}}}$
    is a Jordan automorphism of the Jordan algebra $K_{S_{e}}$.

    On the other hand, if
    $\varphi$ is an arbitrary Jordan automorphism of $K_{S_{e}} \subset M_{3}$,
    then it could be extended to a bistochastic map $S_{h}$ on $M_{3}$,
    for some $h \in \Lambda$,
    by $S_{h} = \varphi \circ S_{e}$.
    Then
    $S_{h} S_{e} = S_{h} = S_{e} S_{h}$,
    because $S_{e}$ acts as the identity map on $K_{S_{e}}$.
    Hence, $he = eh = h$.
    In addition, since $\varphi$ is invertible on $K_{S_{e}}$,
    by extending $\varphi^{-1}$ to another bistochastic map
    $S_{h'}$,  $h' \in \Lambda$,
    we show that $h' \in H(e)$ and $h' h = h h' = e$,
    which proves that $h \in H(e)$.
\end{proof}

\begin{remark}
From the above theorem,
since $\mathcal{G}(\Lambda) = H(\mathbf{1}_{8})$,
 we infer that
the group $\mathcal{G}(\Lambda)$ consists of those matrices that represent
Jordan isomorphism on $M_{3}$,
i.e. for any $g \in \mathcal{G}(\Lambda)$,
there is a unitary matrix $U \in \text{SU}(3)$ such that
either $S_{g}(A) = U A U^{*}$ or $S_{g}(A) = U A^{t} U^{*}$
for all $A \in M_{3}$.
\end{remark}

\begin{theorem}
\label{thm:Idempotents}
The set $\mathcal{E}(\Lambda)$
    is a sum of seven disjoint subsets:
    \begin{equation}
        \mathcal{E}(\Lambda) = \bigcup \limits_{e_{0} \in J} \mathcal{E}_{G_{3}}(e_{0}),
    \end{equation}
    where $J = \left \{ \hat{0}, P_{8}, P_{38}, P_{138}, P_{1238}, P_{13468}, \mathbf{1}_{8} \right \}$.
\end{theorem}
\begin{proof}
    If $e \in \mathcal{E}(\Lambda)$, then as above, $S_{e}$ is a conditional
    expectation map onto the Jordan algebra $K_{S_{e}}$.
    By Theorems 5.3.8 and 6.2.3 of \cite{hanche1984jordan},
    all Jordan subalgebras of $M_{3}$ are isomorphic (unitary equivalent)
    to one of the following:
    $\mathbb{C}\mathbf{1}$,
    $\mathbb{C} E_{12} \oplus \mathbb{C} E_{3}$,
    $\mathbb{C} E_{1} \oplus \mathbb{C} E_{2} \oplus \mathbb{C} E_{3}$,
    $M_{2}^{s} \oplus \mathbb{C} E_{3}$,
    $M_{2} \oplus \mathbb{C} E_{3}$,
    $M_{3}^{s}$,
    and
    $M_{3}$ itself;
    where $M_{k}^{s}$ is the Jordan algebra of symmetric matrices of
    size $k$: $M_{k}^{s} = \{ A \in M_{k}: A = A^{t} \}$;
    $E_{i}, i = 1,2,3$, are matrix units with 1 at the \emph{i}th diagonal entry
    and 0 elsewhere, and $E_{12} = E_{1} + E_{2}$.
    Hence, there is $g \in G_{3}$ such that
    $e = g e_{0} g ^{t}$, and $e_{0}$ is the orthogonal projection
    that represents the projection map onto precisely one of the algebras
    mentioned above.
    It is straightforward to check that then $e_{0} \in J$,
    and $\text{dim} \, e_{0} + 1$ is equal to the dimension of the respective
    Jordan algebra associated to it.
    Hence, $e \in \mathcal{E}_{G_{3}}(e_{0})$.
\end{proof}

\begin{remark}
    \label{rem:noRank6or7}
    In the light of the proof of the theorem above,
    since the are no Jordan subalgebras of $M_{3}$
    of dimension $7$ or $8$,
    nor there are idempotent elements of $\Lambda$
    that have rank $6$ or $7$.
\end{remark}

\begin{corollary}
\label{cor:q}
    For $e \in \mathcal{E}(\Lambda)$, since $e$ is a projection,
    $\text{dim} \, e
    \in \left \{ 0,1,2,3,4,5,8 \right \}$.
\end{corollary}

We prove a useful decomposition of elements of $Q(e)$ in the following lemma.

\begin{lemma}
\label{lem:decomposition}
    Let $e \in \mathcal{E}(\Lambda)$.
    A matrix $x$ belongs to $\in Q(e)$, if and only if,
    $x = h + y$, where
    $h \in H(e)$, $H(e) y = y H(e) = \hat{0}$
    and $\lim_{k \rightarrow \infty} y^{k} = \hat{0}$.
    This decomposition is unique.
\end{lemma}
\begin{proof}
    Suppose that $x = h + y \in \Lambda$, $h \in H(e)$, $y^{k}  \overset{k}{\rightarrow}  \hat{0}$,
    and $hy = yh = \hat{0}$.
    Because $h \in H(e)$, there is a sequence $n_{k} \in \mathbb{N}$ for $k =1,2,3,\ldots$,
    such that $h^{n_{k}} \overset{k}{\rightarrow} e$.
    Then $x^{n_{k}}  = h^{n_{k}} + y^{n_{k}}  \overset{k}{\rightarrow}  e$, i.e.
    $x \in Q(e)$.

    On the other hand,
    let us suppose that $x \in Q(e) \subset \Lambda$.
    We have
    $x = e x e + e^{\perp} x e + e x e^{\perp} + e^{\perp} x ^{\perp}$,
    where $e^{\perp} = \mathbf{1}_{8} - e$.
    By Proposition \ref{prop:SchwarzSemigroups},
    both $x e, x e \in H(e)$.
    Hence, $e x e = x e = e x$, which means that $e^{\perp} x e = \hat{0}$,
    and also $e x e^{\perp} = \hat{0}$.
    Let $h = e x e$, and $y = e^{\perp} x e^{\perp}$.
    We have that $x = h + y$ and $h y = y h = \hat{0}$.
    Then $h^{k} = e x^{k} e$, for $k \in \mathbb{N}$,
    which implies that $h \in Q(e)$.
    Since
    $h =  h e = e h$,
    $h \in H(e)$, again by Proposition \ref{prop:SchwarzSemigroups}.
    For $h' \in H(e)$, we have
    $h' y = h' e y = \hat{0} = y e h' = yh'$,
    i.e. $H(e) y = y H (e) = \hat{0}$.
    There is a sequence $n_{k}$ of natural numbers such that
    $x^{n_{k}}  \overset{k}{\rightarrow}  e$.
    Hence,
    $h^{n_{k}} = e x^{n_{k}} e  \overset{k}{\rightarrow} e^{3}  = e$.
    It follows that $y^{n_{k}}  \overset{k}{\rightarrow}  \hat{0}$,
    which is enough to say that $y^{k} \overset{k}{\rightarrow} \hat{0}$.
    Lastly,
    suppose that $x = h + y = h_{2} + y_{2}$,
    where $h_{2} \in H(e)$,
    and $y_{2}$ has the property:
    $h_{2} y_{2} = y_{2} h_{2} = \hat{0}$.
    For some sequence of natural numbers $m_{k}$, we have
    $e y_{2} = \lim_{k} h_{2}^{m_{k}} y_{2} = \hat{0}$.
    Then $h = e x = e h_{2} + e y_{2} = h_{2}$, and also $y = y_{2}$.
\end{proof}

For $x \in Q(e)$, since the decomposition described above is unique,
let us denote by $h_{x}$ and $y_{x}$ the matrices such that
$x = h_{x} + y_{x}$,
$h_{x} \in H(e)$, $H(e) y_{x} = y_{x} H(e) = \hat{0}$, and
$y_{x}^{k} \overset{k}{\rightarrow} \hat{0}$.
The above lemma justifies the following definition.

\begin{definition}
    Let $e \in \mathcal{E}(\Lambda)$.
    We define:
    $Q_{0}(e) = \left \{ x \in Q(e) \: | \: || y_{x} || < 1 \right \}$.
    For $i = 1, 2, 3, \ldots 8$,
    let $Q_{i}(e)$ be a set consisting of those $x = h_{x} + y_{x} \in Q(e)$,
    for which the largest singular value of $y_{x}$ is equal to 1
    with multiplicity $i$.
\end{definition}

It is obvious that
$Q(e) = \bigcup \left \{ Q_{i}(e), i=1,2,\ldots,8 \right \} \cup Q_{0}(e)$,
and these sets are disjoint (possibly empty).

\begin{proposition}
Let $e \in \mathcal{E}(\Lambda)$.
If $\mathrm{dim} \, e \in \left \{ 5,8 \right \}$,
then
$Q(e) = Q_{0}(e)$.
If $\mathrm{dim} \, e \leq 4$,
then
$Q_{i}(e) = \emptyset$ for $i \geq 5 - \mathrm{dim} \, e$.
\end{proposition}
\begin{proof}
    If $\mathrm{dim} \, e = 8$,
    then $e = \mathbf{1}_{8}$, and since by Remark \ref{rem:Qof1},
    $Q(\mathbf{1}_{8}) = H(\mathbf{1}_{8}) = \mathcal{G}(\Lambda)$,
    obviously $Q(\mathbf{1}_{8}) = Q_{0}(\mathbf{1}_{8})$.
    Suppose that $e \neq \mathbf{1}_{8}$
    and $x = h_{x} + y_{x} \in Q(e)$.
    Then $x^{t} x = h_{x}^{t} h_{x} + y_{x}^{t} y_{x} = e + y_{x}^{t} y_{x}$,
    by Proposition \ref{prop:htranspose};
    $y_{x} e = e y_{x} = \hat{0}$,
    and because for any $k \in \mathbb{N}$, $(x^{t} x)^{k} \in \Lambda$,
    we have that $e + p = \lim_{k} x^{t} x \in \Lambda$,
    where $p$ is a orthogonal projection onto the space spanned by
    eigenvectors of $y_{x}^{t} y_{x}$ with eigenvalue 1.
    Of course, $e p = p e = \hat{0}$, so
    the matrix $e + p \in \mathcal{E}(\Lambda)$.
    It must be that $e + p \neq \mathbf{1}_{8}$,
    otherwise $x^{t} x = \mathbf{1}_{8}$, and hence
    $x \in \mathcal{G}(\Lambda) = Q(\mathbf{1}_{8})$, a contradiction.
    If $\mathrm{dim}\,e = 5$, then $p = \hat{0}$,
    because by Theorem \ref{thm:Idempotents},
    there are no idempotent elements of $\Lambda$
    with rank $6$ or $7$.
    Hence $||y|| < 1$ and $x \in Q_{0}(e)$.
    By the same argument, for $\textrm{dim}\, e \leq 4$,
    it is impossible that $\mathrm{dim} \, p  + \mathrm{dim}\, e \geq 5$,
    so $Q_{i}(e) = \emptyset$, for $i \geq 5 - \mathrm{dim} \, e$,
    because by definition, $i = \textrm{dim}\,p$.
\end{proof}

For the sake of convenience, let us introduce the following sequence
of elements of $\mathcal{E}(\Lambda)$:
$p_{0} = \hat{0}, p_{1} = P_{8}$, $p_{2} = P_{38}$,
$p_{3} = P_{138}$, $p_{4} = P_{1238}$, $p_{5} = P_{13468}$.

\begin{theorem}
\label{thm:LowerIndices}
    Let $i,j$  be integers such that
    $1 \leq i \leq 4$, $0 \leq j \leq 4$ and $i+j \leq 5$.
    If $x \in Q_{i}(p_{j})$,
    then there exist $g_{1}, g_{2} \in G_{3}$ and $z \in Q_{0}(p_{i+j})$
    such that $x = g_{1} z g_{2}$.
\end{theorem}
\begin{proof}
    Let $x \in Q_{i}(p_{j})$ and
    $x = h_{x} + y_{x}$, as above.
    Since this decomposition is unique, we can write as in the proof of
    Lemma \ref{lem:decomposition}:
    $h_{x} = p_{j} x p_{j}$, and $y_{x} = p_{j}^{\perp} x p_{j}^{\perp}$,
    where $p_{j}^{\perp} = \mathbf{1}_{8} - p_{j}$.
    Because $p_{j} = p_{j}^{t}$, then
    $h y^{t} = y h^{t} = \hat{0}$, and thus
    $x x^{t}  = h h^{t} + y y^{t} = p_{j} + y y^{t}$,
    by \mbox{Proposition \ref{prop:htranspose}}.
    Let $R_{1} (p_{i} + y_{0}) R_{2}$ be the singular value decomposition of $y$,
    i.e. $R_{1}, R_{2}  \in \text{O}(8)$ are orthogonal matrices,
    and $y_{0}$ is diagonal with the only possible non-zero entries $s_{1}, s_{2}, \ldots, s_{8-i}$,
    such that $1 > s_{1} \geq s_{2} \geq \ldots \geq s_{8-i} \geq 0$,
    and $p_{i} y_{0} = y_{0} p_{i} = \hat{0}$.
    Then for $k \in \mathbb{N}$,
    $(x x^{t})^{k} = p_{j} + R_{1}( p_{i} + (y_{0} y_{0}^{t})^{k}) R_{1}^{t} \in \Lambda$,
    and  because $\Lambda$ is closed:
    $e_{1} = p_{j} + R_{1} p_{i} R_{1}^{t} = \lim_{k} (x x^{t})^{k} \in \Lambda$.
    Also, since $h y^{t} = y h^{t} = \hat{0}$,
    we have that $p_{j} R_{1} p_{i} R_{1}^{t} = R_{1} p_{i} R_{1}^{t} p_{j} = \hat{0}$.
    It follows that $e_{1}$ is an idempotent and $\text{rank} \, e_{1} = i+j$.
    By Theorem \ref{thm:Idempotents}, there is $g_{1} \in G_{3}$ such that
    $e_{1} = g_{1} p_{i+j} g_{1}^{t}$.
    A similar argument,
    applied this time to $x^{t} x$, shows that there is $g_{2} \in G_{3}$
    such that the idempotent $e_{2} = p_{j} + R_{2}^{t} p_{i} R_{2} \in \Lambda$
    could be written as $e_{2} = g_{2}^{t} p_{i+j} g_{2}$.
    Let $z = g_{1}^{t} x g_{2}^{t}$.
    What remains to show is that $z \in Q_{0}(p_{i+j})$.
    It is evident that $z \in \Lambda$.
    One can easily check that
    $p_{i+j} z = z p_{i+j}$,
    and hence $z = h_{z} + y_{z}$,
    where $h_{z} = p_{i+j} z p_{i+j}$,
    and $y_{z} = p_{i+j}^{\perp} z p_{i+j}^{\perp}$,
    $p_{i+j}^{\perp} = \mathbf{1}_{8} - p_{i+j}$.
    Because $h_{z} = p_{i+j} h_{z} = h_{z} p_{i+j}$,
    $h_{z} \in H(p_{i+j})$.
    Obviously, $h y_{z} = y_{z} h = \hat{0}$, for any $h \in H(p_{i+j})$.
    In addition, we have
    $y_{z} = g_{1}^{t}R_{1} y_{0} R_{2} g_{2}^{t}$,
    so $||y_{z}^{k}|| \leq ||y_{0}||^{k} = s_{1}^{k} \overset{k}{\rightarrow} \hat{0}$.
    Therefore, by Lemma \ref{lem:decomposition},
    $z \in Q(p_{i+j})$.
    Since $||y_{z}|| < 1$,
    $z \in Q_{0}(p_{i+j})$, which ends the proof.
\end{proof}

The main result of this paper could be captured in the following remark.
\begin{corollary}
\label{cor:q0}
    Suppose that $x \in \text{Ext}_{0}(\Lambda)$ and $S_{x}$ is not a Jordan isomorphism.
    Then there exist $g_{1}, g_{2} \in G_{3}$
    such that $g_{1}^{t} x g_{2}^{t} \in Q_{0}(\hat{0}) \cup Q_{0}(P_{8})$.
    In other words,
    either $||x|| < 1$,
    or
    $x = g_{1} (P_{8} + y) g_{2}$,  $y P_{8} = P_{8} y = \hat{0}$,
    and $|| y || < 1$.
\end{corollary}
\begin{proof}
Let $x \in  \text{Ext}_{0}(\Lambda)$.
    By Lemma \ref{lem:equivClassesOfQe} and Theorem \ref{thm:Idempotents}, there is $g \in G_{3}$ such that
    $g^{t} x g \in Q(e_{0})$, and $e_{0} \in J$.
    Because $S_{x}$ is an extremal positive map, then by
    \cite[Theorem 2]{miller2015stable},
    $e_{0} \in \left \{ \hat{0}, P_{8}, \mathbf{1}_{8} \right \}$.
    By Remark \ref{rem:Qof1}, $x \notin Q(\mathbf{1}_{8})$, hence
    $z = g^{t} x g \in Q(\hat{0}) \cup Q(P_{8})$,
    and $z \in \text{Ext}_{0}(\Lambda)$
    (cf. \cite[Lemma 3.1.2]{stormer2013positive}).
    Suppose that $z \in Q(P_{8})$.
    By Theorem \ref{thm:LowerIndices},
    because $z \in \text{Ext}_{0}(\Lambda)$,
    $z \notin Q_{i}(P_{8})$
    for $i \geq 1$ (see  \cite[Theorem 2]{miller2015stable}).
    Hence $z \in Q_{0}(P_{8})$, and put $g_{1} = g$ and $g_{2} = g^{t}$.
    Then $x = g_{1} z g_{2}$,
    and by Lemma \ref{lem:decomposition}, because $H(P_{8}) = \left \{ P_{8} \right \}$
    (see Remark \ref{rem:GOfLambda}),
    $z = P_{8} + y$,  $y P_{8} = P_{8} y = \hat{0}$, and $||y|| < 1$.
    On the other hand, if $z \in Q(\hat{0})$,
    then by the same reasoning either $z \in Q_{0}(\hat{0})$,
    and $||z|| < 1$, because $H(\hat{0}) = \left \{ \hat{0} \right \}$, or
    $z \notin Q_{0}(\hat{0})$ and then there are $g_{01}, g_{02} \in G_{3}$ such that
    $g_{01} z g_{02} \in Q_{0}(P_{8})$.
    Then put $g_{1} = g g_{01}^{t}$ and $g_{2} = g_{02}^{t} g^{t}$,
    and the assertion follows.
\end{proof}

We summarise the results presented above by saying that they allow
to narrow the task of finding positive extremal and bistochastic maps
on $M_{3}(\mathbb{C})$ to three specific groups;
and examples in all these groups have been found previously.
First, there are Jordan isomorphisms, represented by matrices from
$Q(\mathbf{1}_{8}) = \mathcal{G}(\Lambda)$.
Second, there are maps that could be called \emph{strongly ergodic}
\cite{miller2015stable},
belonging to the class represented by
$\left \{ x \in Q_{0}(\hat{0}) \:|\: ||x|| = \frac{1}{2} \right \}$,
of which the celebrated Choi map
\cite{choi1977extremal} is an example.
For a generalised Choi map $\Phi[a,b,c]$:
\begin{equation}\label{eq:choi}
\Phi[a,b,c](X)=\\
\frac{1}{2}
\begin{pmatrix}
ax_{11}+bx_{22}+cx_{33} & -x_{12} & -x_{13} \\
-x_{21} & cx_{11}+ax_{22}+bx_{33} & -x_{23} \\
-x_{31} & -x_{32} & bx_{11}+cx_{22}+ax_{33}
\end{pmatrix},
\end{equation}
where $X = (x_{ij})_{i,j = 1}^{3}$,
if we parametrise:
\begin{equation}
a(t)=\dfrac{(1-t)^2}{1-t+t^2},\quad b(t)=\dfrac{t^2}{1-t+t^2},\quad c(t)=\dfrac 1{1-t+t^2},
\end{equation}
$0 \leq t < 1$,
then $\Phi[a(t), b(t), c(t)]$ is an extremal bistochastic map on $M_{3}$.
For $t = 0$, it is the map proposed originally by Choi,
and for $t = 1$ it is a completely positive map.
See the paper by
K.-C.\,Ha and S.-H.\,Kye \cite{ha2011entanglement} for more details.
The family of matrices $x_{t}$ such that
$S_{x_{t}} = \Phi[a(t), b(t), c(t)]$ is given by
\begin{equation}
    x_{t} = \begin{pmatrix}
        - \frac{1}{2} & 0 & 0 & 0 & 0 & 0 & 0 & 0 \\
        0 & - \frac{1}{2} & 0 & 0 & 0 & 0 & 0 & 0 \\
        0 & 0 & \frac{1 - 4t + t^{2}}{4(1 - t + t^{2})} & 0 & 0 & 0 & 0 & - \frac{\sqrt{3}(1 - 4t + t^{2})}{4(1 - t + t^{2})} \\
        0 & 0 & 0 & - \frac{1}{2} & 0 & 0 & 0 & 0 \\
        0 & 0 & 0 & 0 & - \frac{1}{2} & 0 & 0 & 0 \\
        0 & 0 & 0 & 0 & 0 & - \frac{1}{2} & 0 & 0 \\
        0 & 0 & 0 & 0 & 0 & 0 & - \frac{1}{2} & 0 \\
        0 & 0 & \frac{\sqrt{3}(1 - 4t + t^{2})}{4(1 - t + t^{2})} & 0 & 0 & 0 & 0 & \frac{1 - 4t + t^{2}}{4(1 - t + t^{2})} \\
    \end{pmatrix}.
\end{equation}
It is worth noting that each
$x_{t} = \frac{1}{2} R_{t}$, where $R_{t} \in \text{O}(8)$,
as in Proposition \ref{prop:oneHalfofOrthogonal}.

Lastly, there are maps represented by elements of $Q_{0}(P_{8})$,
one example of which was proposed in \cite{miller2015stable}:
\begin{equation}
\label{eq:DefinitionOfS}
S_{0}(X) \:=\: \begin{pmatrix}
        \frac{1}{2}(x_{11} + x_{22}) & 0 & \frac{1}{\sqrt{2}} x_{13} \\
        0 & \frac{1}{2}(x_{11} + x_{22}) & \frac{1}{\sqrt{2}} x_{32} \\
        \frac{1}{\sqrt{2}} x_{31} & \frac{1}{\sqrt{2}} x_{23} & x_{33}
        \end{pmatrix}.
\end{equation}
The matrix $x \in \Lambda$ such that $S_{x} = S_{0}$ is given by
a diagonal matrix
\begin{equation}
 x =
\text{diag}(0,0,0,
\frac{1}{\sqrt{2}}, \frac{1}{\sqrt{2}}, \frac{1}{\sqrt{2}}, \frac{-1}{\sqrt{2}},
1).
\end{equation}

We say that two elements $x, y \in \Lambda$
are equivalent if, and only if, there are $g_{1}, g_{2} \in G_{3}$
such that $x = g_{1} y g_{2}$.
Then the task of finding elements in $\text{Ext}_{0}(\Lambda)$,
up to this equivalence relation,
consists of:
\begin{enumerate}
\item finding all $R \in \mathrm{O}(8)$ such that
    $\frac{1}{2}R \in \text{Ext}_{0}(\Lambda)$;
\item determining which elements $x \in \Lambda$,
    such that $\frac{1}{2} < ||x|| < 1$, belong to $\text{Ext}_{0}(\Lambda)$;
\item finding all $y$, such that
    $y P_{8} = P_{8} y = \hat{0}$, $||y|| < 1$,
    and $P_{8} + y \in \text{Ext}_0(\Lambda)$.
\end{enumerate}

In this paper,
we have identified the \emph{bistochastic} maps of $M_{3}$
with a subset $\Lambda \subset M_{8}(\mathbb{R})$,
and studied its properties as a compact and convex semigroup.
In applications to the quantum information theory,
one would ideally seek the classification of \emph{all} positive maps,
not only the ones that preserve trace and identity.
The assumption that a positive map has only one of those properties,
i.e. it preserves either the trace or the identity,
can be taken without loss of the generality.
To assume, however, that a map is bistochastic does present a constraint
to a certain degree.
The semigroup of all positive maps that preserve, say, only
the identity matrix,
could be identified with a compact and convex semigroup
within the product space
$M_{8}(\mathbb{R}) \! \times \!  \mathbb{R}^{8}$,
together with the group multiplication that comes from
a semi-direct product of $M_{8}(\mathbb{R})$ and $\mathbb{R}^{8}$.
This semigroup is almost certainly much more difficult to study,
and we decided to start our investigation from a simpler case
of the semigroup $\Lambda$.
In the future,
we hope that the methods proposed,
together with some other, mostly geometric techniques,
will allow to generalise the above result to a higher dimensional case and
to bring the research closer to the final classification
of extremal positive maps on matrix algebras.



\bibliographystyle{abbrv}
\bibliography{./biblio}

\begin{thebibliography}{10}

\bibitem{abramovich2002problems}
Y.~A. Abramovich and C.~D. Aliprantis.
\newblock {\em Problems in operator theory}, volume~51.
\newblock American Mathematical Society, 2002.

\bibitem{choi1980some}
M.-D. Choi.
\newblock {Some assorted inequalities for positive linear maps on C*-algebras}.
\newblock {\em {Journal of Operator Theory}}, 4(2):271--285, 1980.

\bibitem{choi1977extremal}
M.-D. Choi and T.-Y. Lam.
\newblock Extremal positive semidefinite forms.
\newblock {\em {Mathematische Annalen}}, 231(1):1--18, 1977.

\bibitem{chow1975compact}
H.~Chow.
\newblock On compact affine semigroups.
\newblock In {\em Semigroup Forum}, volume~11, pages 146--152. Springer, 1975.

\bibitem{chruscinski2014entanglement}
D.~Chruściński and G.~Sarbicki.
\newblock Entanglement witnesses: construction, analysis and classification.
\newblock {\em Journal of Physics A: Mathematical and Theoretical},
  47(48):483001, 2014.

\bibitem{ha2011entanglement}
K.-C. Ha and S.-H. Kye.
\newblock Entanglement witnesses arising from exposed positive linear maps.
\newblock {\em {Open Systems and Information Dynamics}}, 18(04):323--337, 2011.

\bibitem{hanche1984jordan}
H.~Hanche-Olsen and E.~St{\o}rmer.
\newblock {\em Jordan operator algebras}.
\newblock Pitman Advanced Pub. Program, 1984.

\bibitem{horodecki1996separability}
M.~Horodecki, P.~Horodecki, and R.~Horodecki.
\newblock Separability of mixed states: necessary and sufficient conditions.
\newblock {\em {Physics Letters A}}, 223(1-2):1--8, 1996.

\bibitem{miller2014horodeckis}
M.~Miller and R.~Olkiewicz.
\newblock {Horodeckis criterion of separability of mixed states in von Neumann
  and C*-algebras}.
\newblock {\em {Infinite Dimensional Analysis, Quantum Probability and Related
  Topics}}, 17(04):1450028, 2014.

\bibitem{miller2015stable}
M.~Miller and R.~Olkiewicz.
\newblock {Stable subspaces of positive maps of matrix algebras}.
\newblock {\em Open Systems and Information Dynamics}, 22(02):1550011, 2015.

\bibitem{miller2015topology}
M.~Miller and R.~Olkiewicz.
\newblock Topology of the cone of positive maps on qubit systems.
\newblock {\em Journal of Physics A: Mathematical and Theoretical},
  48(25):255203, 2015.

\bibitem{peres1996separability}
A.~Peres.
\newblock Separability criterion for density matrices.
\newblock {\em {Physical Review Letters}}, 77(8):1413--1415, 1996.

\bibitem{schwarz1955hausdorff}
S.~Schwarz.
\newblock {On Hausdorff bicompact semigroups}.
\newblock {\em Czechoslovak Math. J}, 5(80):1--23, 1955.

\bibitem{stormer1963positive}
E.~St{\o}rmer.
\newblock Positive linear maps of operator algebras.
\newblock {\em {Acta Mathematica}}, 110(1):233--278, 1963.

\bibitem{stormer2013positive}
E.~St{\o}rmer.
\newblock {\em Positive Linear Maps of Operator Algebras}.
\newblock Springer Monographs in Mathematics. Springer, 2013.

\bibitem{woronowicz1976positive}
S.~L. Woronowicz.
\newblock Positive maps of low dimensional matrix algebras.
\newblock {\em {Reports on Mathematical Physics}}, 10(2):165--183, 1976.

\end{thebibliography}

\end{document}